\documentclass[11pt]{article}
\usepackage[dvips]{graphicx}
\usepackage{bm}
\usepackage{amsmath,amsthm,amssymb}
\usepackage{multirow,bigdelim}
\usepackage{amscd,fullpage}
\usepackage[usenames,dvipsnames]{color}
\usepackage[letterpaper=true,colorlinks=true,pdfpagemode=none,urlcolor=blue,linkcolor=blue,citecolor=BrickRed,pdfstartview=FitH]{hyperref}
\usepackage[noend]{algorithmic}
\usepackage[ruled]{algorithm}
\renewcommand{\algorithmicrequire}{\textbf{Input:}}
\renewcommand{\algorithmicensure}{\textbf{Output:}}

\newtheorem{thm}{Theorem}[section]
\newtheorem{theorem}[thm]{Theorem}
\newtheorem{proposition}[thm]{Proposition}
\newtheorem{lemma}[thm]{Lemma}

\newtheorem{claim}[thm]{Claim}

\newcommand{\bs}{{\bm s}}
\newcommand{\bo}{{\bm o}}
\newcommand{\bx}{{\bm x}}
\newcommand{\by}{{\bm y}}
\newcommand{\bt}{{\bm t}}

\newcommand{\calP}{\mathcal{P}}
\newcommand{\bbR}{\mathbb{R}}
\newcommand{\bbZ}{\mathbb{Z}}

\newcommand{\Zivny}{\v{Z}ivn\'y}

\title{Improved Approximation Algorithms for $k$-Submodular Function Maximization}
\author{Satoru Iwata
\thanks{Department of Mathematical Informatics, Graduate School of Information Science and Technology,
  University of Tokyo, Tokyo 113-8656, Japan ({\tt iwata@mist.i.u-tokyo.ac.jp}).
  Supported by JSPS Grant-in-Aid for Scientific Research (B) No.~23300002.}
\and
Shin-ichi Tanigawa\thanks{
  Research Institute for Mathematical Sciences, Kyoto University, Kyoto 606-8502 Japan  (\texttt{tanigawa@kurims.kyoto-u.ac.jp}).
  Supported by JSPS Grant-in-Aid for Scientific Research (B) No.~23300002.}
  \and
Yuichi Yoshida\thanks{
  National Institute of Informatics, and  Preferred Infrastructure, Inc.
  (\texttt{yyoshida@nii.ac.jp}).
    Supported by JSPS Grant-in-Aid for Young Scientists (B) (No.~26730009), MEXT Grant-in-Aid for Scientific Research on Innovative Areas (No.~24106003), and JST, ERATO, Kawarabayashi Large Graph Project.} }
\date{\today}

\begin{document}
\maketitle
\begin{abstract}
This paper presents a polynomial-time $1/2$-approximation algorithm for maximizing nonnegative $k$-submodular functions. This improves upon the previous $\max\{1/3, 1/(1+a)\}$-approximation by Ward and \Zivny~\cite{Ward:2014wja}, where $a=\max\{1, \sqrt{(k-1)/4}\}$. We also show that for monotone $k$-submodular functions there is a polynomial-time $k/(2k-1)$-approximation algorithm while for any $\varepsilon>0$ a $((k+1)/2k+\varepsilon)$-approximation algorithm for maximizing monotone $k$-submodular functions would require exponentially many queries.
In particular, our hardness result implies that our algorithms are asymptotically tight.

We also extend the approach to provide constant factor approximation algorithms for maximizing
skew-bisubmodular functions,
which were recently introduced as generalizations of bisubmodular functions.
\end{abstract}

\section{Introduction}

Let $2^V$ denote the family of
all the subsets of $V$. A function $g:2^V\to\bbR$ is called \emph{submodular} if it satisfies
$$g(Z_1)+g(Z_2)\geq g(Z_1\cup Z_2)+g(Z_1\cap Z_2)$$ for every pair of $Z_1$ and $Z_2$ in $2^V$.
Submodular function maximization contains important NP-hard optimization problems such as max cut and certain facility location problems. It is known to be intractable in the standard value oracle model, and approximation algorithms
have been studied extensively. In particular, Feige, Mirrokni, and Vondr\'ak~\cite{Feige:2011iu} have developed
constant factor approximation algorithms for the unconstrained maximization of nonnegative submodular
functions and shown that no approximation algorithm can achieve the ratio better than $1/2$.
Buchbinder, Feldman, Naor, and Schwartz~\cite{bfns} provided much simpler algorithms that substantially
improve the approximation factor. In particular their randomized version, called the randomized double-greedy algorithm, achieves the factor of $1/2$, which is the best possible in the oracle value model.

%

In this paper we shall consider the maximization of nonnegative $k$-submodular functions which generalizes the submodular function maximization. Let $(k+1)^V:=\{(X_1,\dots, X_k)\mid X_i\subseteq V\, (i=1,\ldots,k), X_i\cap X_j=\emptyset\, (i\neq j)\}$.
 A function $f:(k+1)^V\rightarrow \mathbb{R}$ is called \emph{$k$-submodular} if,
for any $\bx=(X_1,\dots,X_k)$ and $\by=(Y_1,\dots, Y_k)$ in $(k+1)^V$, we have
 \begin{eqnarray*}
 f(\bx)+f(\by) & \geq & f(\bx\sqcup \by)+f(\bx\sqcap \by)
 \end{eqnarray*}
where
\begin{align*}
\bx\sqcap \by&:=(X_1\cap Y_1,\dots, X_k\cap Y_k) \\
 \bx\sqcup\by&:=\left(X_1\cup Y_1\setminus \left(\bigcup_{i\neq 1}X_i\cup Y_i\right), \dots, X_k\cup Y_k\setminus \left(\bigcup_{i\neq k}X_i\cup Y_i\right)\right).
\end{align*}
$k$-submodular functions were first introduced by Huber and Kolmogorov~\cite{hk} as a generalization of {\em bisubmodular functions}, which correspond to $2$-submodular functions in the above notation.
Examples of bisubmodular functions include the rank functions of delta-matroids and the cut capacity functions of bi-directed networks, and the minimization problem has been extensively studied~\cite{fi,mf}.
Examples of $k$-submodular functions will be explained later.

Ward and \Zivny~\cite{Ward:2014wja} and the present authors~\cite{ity} independently observed that algorithms for submodular function maximization due to Buchbinder, Feldman, Naor, and Schwartz~\cite{bfns} can be naturally extended to bisubmodular function maximization. In particular the randomized double greedy algorithm for submodular functions can be seen as a randomized greedy algorithm in the bisubmodular setting and it achieves the best approximation ratio $1/2$.
Ward and \Zivny~\cite{Ward:2014wja} further analyzed the randomized greedy algorithm for $k$-submodular function maximization and proved that its approximation ratio is $1/(1+a)$,
where $a=\max\{1,\sqrt{(k-1)/4}\}$.
They also gave a deterministic $1/3$-approximation algorithm.

In this paper we shall present an improved $1/2$-approximation algorithm for maximizing $k$-submodular functions. Our algorithm follows the randomized greedy framework as in~\cite{ity, Ward:2014wja}, and the main idea is the use of a different probability distribution derived from a geometric sequence at each step.

By extending the argument by Feige, Mirrokni, and Vondr{\'a}k~\cite{Feige:2011iu} we also show that for any $\varepsilon>0$,  a $((k+1)/2k +\varepsilon)$-approximation for the $k$-submodular function
maximization problem would require exponentially many queries, implying the tightness of our result for large $k$.
In fact, our inapproximability result holds for a much restricted class of monotone $k$-submodular functions,
where a $k$-submodular function  is said to be \emph{monotone} if $f(\bx)\leq f(\by)$ for any
$\bx=(X_1,\dots, X_k)$ and $\by=(Y_1,\dots, Y_k)$ in $(k+1)^V$ with $X_i\subseteq Y_i$ for $1\leq i\leq k$.
On the other hand we show that  there is a $k/(2k-1)$-approximation for monotone $k$-submodular functions.
In particular it attains an approximation ratio of $2/3$ for bisubmodular functions.


In order to understand the relation between $k$-submodular function maximization and other maximization problems, it is useful to understand characteristic properties of $k$-submodular functions, called orthant submodularity and pairwise monotonicity.
To see them, define a partial order $\preceq$ on $(k+1)^V$ such that,
for $\bx=(X_1,\dots, X_k)$ and $\by=(Y_1,\dots, Y_k)$ in $(k+1)^V$,
$\bx\preceq \by$  if $X_i\subseteq Y_i$ for every $i$ with $1\leq i\leq k$.
Also,  define
\[
\Delta_{e,i}f(\bx)=f(X_1,\dots, X_{i-1}, X_i\cup \{e\}, X_{i+1},\dots, X_k)-f(X_1,\dots, , X_k)
\]
for $\bx\in (k+1)^V$, $e\notin \bigcup_{j=1}^k X_j$, and $i\in [k]$, which is a marginal gain when adding $e$ to the $i$-th component of $\bx$.
Then it is easy to see that the $k$-submodularity implies
the \emph{orthant submodularity}:
\begin{equation*}
\Delta_{e,i}f(\bx)\geq \Delta_{e,i}f(\by) \qquad
(\text{$\bx, \by\in (k+1)^V$ with $\bx\preceq \by$, $e\notin \bigcup_{j\in [k]} Y_j$,  and $i\in [k]$}),
\end{equation*}
and the \emph{pairwise monotonicity}:
\begin{equation*}
\Delta_{e,i} f(\bx)+\Delta_{e,j} f(\bx)\geq 0 \qquad (\text{$\bx\in (k+1)^V$, $e\notin \bigcup_{\ell\in [k]} X_{\ell}$, and  $i,j\in [k]$ with $i\neq j$}).
\end{equation*}
Ward and \Zivny~\cite{Ward:2014wja} showed that these properties indeed characterize $k$-submodular functions, extending the corresponding result for bisubmodular functions~\cite{ando1996characterization}.
\begin{theorem}[Ward and \Zivny~\cite{Ward:2014wja}]
\label{thm:char}
A function $f:(k+1)^V\rightarrow \mathbb{R}$ is $k$-submodular if and only if $f$ is orthant submodular and pairwise monotone.
\end{theorem}

The $k$-submodular function maximization problem is closely related to the submodular function maximization with a partition matroid constraint. Consider a partition $\{U_1,\dots, U_n\}$ of a finite set $U$ such that $|U_i|=k$ and a partition matroid on $U$ such that $I\subseteq U$ is independent if and only if $|I\cap U_i|\leq 1$ for every $1\leq i\leq n$.
By identifying each $U_i$ with $[k]$, one can identify each independent set $I$ with an element $\bx$ of $(k+1)^V$, where $V=\{1,\dots, n\}$.
Therefore, for a given submodular function $g:2^U\rightarrow \mathbb{R}_+$, its restriction to the family of independent sets can be considered as a function from $(k+1)^V$ to $\mathbb{R}_+$ satisfying orthant submodularity.
In general, if $g$ is monotone, the submodular function maximization with a matroid constraint admits $(1-1/\mathrm{e})$-approximation~\cite{calinescu}, which is known to be best possible in the value oracle model~\cite{mirrokni2008tight}.
On the other hand, when $g$ is non-monotone, the current best approximation ratio is $1/\mathrm{e}$~\cite{feldman2011unified} for general matroids, and deriving the tight bound is recognized as a challenging problem even for uniform matroids (see~\cite{bfns14}).
The $k$-submodular function maximization is in between:
it admits $1/2$-approximation whereas it assumes pairwise monotonicity, which is strictly weaker than monotonicity.

It is also worth mentioning that in the $k$-submodular function maximization there always exists a maximizer which is a partition of $V$ (c.f.~Proposition~\ref{prop}), which corresponds with a base in the partition matroid.
Vondr{\'a}k~\cite{Vondrak:2009hr} showed that, under a matroid base constraint,  any $(1-1/\nu+\varepsilon)$-approximation requires exponentially many queries for any $\varepsilon>0$, where $\nu$ denotes the fractional packing number (see~\cite{Vondrak:2009hr} for the definition). One can easily show that $\nu=k$ in our case, and hence this general result does not give a nontrivial bound for large $k$.

We should also remark that, in the $k$-submodular function maximization problem, function values are determined over $(k+1)^V$ and hence over the independent sets in the corresponding submodular function maximization with a partition matroid constraint. It is not in general true that such a non-negative (monotone) function can be extended to a non-negative (monotone) submodular function over $2^U$.

An important special case of the submodular function maximization with a partition matroid constraint  is the submodular welfare problem.
In the submodular welfare problem, given a finite set $V$ and \emph{monotone} submodular functions $g_i:2^V\rightarrow \mathbb{R}_+$ for $1\leq i\leq k$, we are asked to find a partition $\{X_1,\dots, X_k\}$ of $V$ that maximizes $\sum_{i=1}^k g_i(X_i)$. Feldman, Naor and Schwartz~\cite{feldman2011unified} gave a $(1-(1-1/k)^k)$-approximation approximation algorithm, which is known to be best possible in the value oracle model~\cite{Vondrak:2009hr}.
Now,  consider $h:(k+1)^V\rightarrow \mathbb{R}_+$ given by
\[
h(X_1,\dots, X_k)=\sum_{i=1}^k g_i(X_i) \qquad ((X_1,\dots, X_k)\in (k+1)^V).
\]
Then the submodularity and the monotonicity of $g_i$ imply the orthant submodularity and the pairwise monotonicity of $h$, and hence $h$ is monotone $k$-submodular by Theorem~\ref{thm:char}.
Thus the monotone $k$-submodular function maximization generalizes the submodular welfare problem.
In fact we will show that the approximation algorithm by Dobzinski and Schapira~\cite{Dobzinski:2006wi} for the submodular welfare problem can be extended to the monotone case.

A similar construction gives another interesting application of the $k$-submodular function maximization.
For a submodular function $g:2^V\rightarrow \mathbb{R}_+$, define $h':(k+1)^V\rightarrow \mathbb{R}_+$ by
\[
 h'(X_1,\dots, X_k)=\sum_{i=1}^k g(X_i) \qquad ((X_1,\dots, X_k)\in (k+1)^V).
\]
The resulting $h'$ satisfies orthant submodularity but may not satisfy pairwise monotonicity in general.
However if $g$ is symmetric (i.e., $g(X)=g(V\setminus X)$ for $X\subseteq V$) it turns out that $h'$ is pairwise monotone and thus it is $k$-submodular by Theorem~\ref{thm:char}.
Therefore, for a symmetric submodular function $g$, our algorithm gives a $\frac{1}{2}$-approximation for the problem of finding a partition $\{X_1,\dots, X_k\}$ of $V$ that maximizes $\sum_{i=1}^k g(X_i)$.
Note that this problem generalizes the Max $k$-cut problem.

As another extension of the bisubmodularity, Huber, Krokhin, and Powell~\cite{hkp} have introduced the concept of
skew-bisubmodularity.
For $\alpha\in [0,1]$, a function $f:3^V \to \mathbb{R}$ is called \emph{$\alpha$-bisubmodular} if,
for any $\bx=(X_1,X_2)$ and $\by=(Y_1,Y_2)$ in $3^V$,
\begin{align*}
f(\bx)+f(\by)\geq f(\bx\sqcap \by)+\alpha f(\bx\sqcup \by)+(1-\alpha) f(\bx \dot\sqcup \by),
\end{align*}
where
\[
 \bx \dot\sqcup \by=(X_1\cup Y_1,(X_2\cup Y_2)\setminus (X_1\cup Y_1)).
\]
A function $f:3^V \to \mathbb{R}$ is called \emph{skew-bisubmodular}
if it is $\alpha$-bisubmodular for some $\alpha\in [0,1]$.


We show that a randomized greedy algorithm provides an approximate solution within the factor
of $\frac{2\sqrt{\alpha}}{(1+\sqrt{\alpha})^2}$ for maximizing an $\alpha$-bisubmodular function.
 This means that the double greedy algorithm of Buchbinder~\emph{et~al.}~\cite{bfns} relies on a symmetry of submodular functions.
Combining this with another simple algorithm, we obtain an approximate algorithm whose approximate
ratio is at least $\frac{8}{25}$ for any $\alpha\in [0,1]$.
This result has been included in our previous technical report~\cite{ity}, but not in a reviewed article.

The rest of this paper is organized as follows. In Section~\ref{sec:k-submodular}, we present our approximation algorithms for the $k$-submodular function maximization.
In Section~\ref{sec:inapprox}, we discuss the inapproximability.
In Section~\ref{sec:a-bisubmodular} we analyze a randomized greedy algorithm for maximizing $\alpha$-bisubmodular functions, and then we present an improvement that leads to a constant-factor approximation algorithm.

\section{Approximation algorithms for $k$-submodular functions}
\label{sec:k-submodular}
In this section we give approximation algorithms for the $k$-submodular function maximization problem.
To analyze $k$-submodular functions it is often convenient to identify $(k+1)^V$ as $\{0,1\dots, k\}^V$,
that is, the set of $|V|$-dimensional vectors with entries in $\{0,1,\dots, k\}$.
Namely, we associate  $(X_1,\dots, X_k)\in (k+1)^V$ with $\bx\in \{0,1,\dots, k\}^V$ by
$X_i=\{e\in V\mid \bx(e)=i\}$ for $1\leq i\leq k$.
Hence we sometimes abuse notation, and simply write $\bx=(X_1,\dots, X_k)$ by regarding a vector $\bx$ as a subpartition of $V$.

For $\bx\in \{0,1,\dots, k\}^V$, let ${\rm supp}(\bx)=\{e\in V\mid x(e)\neq 0\}$,
and let ${\bf 0}$ be the zero vector in $\{0,1,\dots, k\}^V$.

\subsection{Framework}

Our approximation algorithms are obtained from the following meta-framework (Algorithm~\ref{alg:meta}) for maximizing $k$-submodular functions by changing the probability distributions used in the framework.

\begin{algorithm}
  \caption{}\label{alg:meta}
  \begin{algorithmic}
  \REQUIRE{A nonnegative $k$-submodular function $f:\{0,1,\ldots,k\}^V \to \bbR_+$.}
  \ENSURE{A vector $\bs$.}
  \STATE{$\bs\leftarrow0$.}
  \FOR{each $e \in V$}
    \STATE{Set a probability distribution $p$ over $\{1,\dots, k\}$.}
    \STATE{Let $\bs(e) \in \{1,\ldots,k\}$ be chosen randomly, with $\Pr[\bs(e) = i] = p_i$ for all $i \in \{1,\ldots,k\}$.}
  \ENDFOR{}
  \RETURN{$\bs$}
  \end{algorithmic}
\end{algorithm}
The approximation algorithms  for bisubmodular functions~\cite{ity} and more generally for $k$-submodular functions~\cite{Ward:2014wja} are specializations of Algorithm~\ref{alg:meta}, where the probability distribution is chosen to be proportional to its marginal gain.

We now evaluate the quality of the solution of Algorithm~\ref{alg:meta} by applying the analysis in~\cite{ity, Ward:2014wja}.
We first remark the following key fact (see \cite{ity, Ward:2014wja} for the proof).
\begin{proposition}
\label{prop}
For any $k$-submodular function $f:(k+1)^V\rightarrow \mathbb{R}_+$, there exists a partition of $V$ that attains the maximum value of $f$.
\end{proposition}
We also need the following notation, which will be used throughout this section.
Let $n=|V|$.
By Proposition~\ref{prop} there is an optimal solution $\bo$ with ${\rm supp}(\bo)=V$.
Let $\bs$ be the output of the algorithm.
We consider the $j$-th iteration of the algorithm, and let
$e^{(j)}$ be the element of $V$ considered in the $j$-th iteration,
$p_i^{(j)}$ be the probability that $i$-th coordinate is chosen in the $j$-th iteration, and
$\bs^{(j)}$ be the solution after the $i$-th iteration, where $\bs^{(0)}={\bf 0}$.
Also for $0\leq j\leq n$ let $\bo^{(j)}=(\bo \sqcup \bs^{(j)})\sqcup \bs^{(j)}$, that is, the element in $\{0,1,\dots, k\}^V$ obtained from $\bo$ by replacing the coordinates on ${\rm supp}(\bs^{(j)})$ with those of $\bs^{(j)}$,
and for $1\leq j\leq n$ let $\bt^{(j-1)}=(\bo \sqcup \bs^{(j)})\sqcup \bs^{(j-1)}$, that is, the one obtained from $\bo^{(j)}$ by changing
$\bo^{(j)}(e^{(j)})$ with $0$.
Also for $i\in [k]$ let $y_i^{(j)}=\Delta_{e^{(j)}, i}f(\bs^{(j-1)})$
and let $a_i^{(j)}=\Delta_{e^{(j)}, i}f(\bt^{(j-1)})$.
Due to the pairwise monotonicity, we have
\begin{align}
\label{01}
y_i^{(j)}+y_{i'}^{(j)}&\geq 0 \qquad (i, i' \in [k], i\neq i'), \\ \label{02}
a_i^{(j)}+a_{i'}^{(j)}&\geq 0 \qquad (i, i' \in [k], i\neq i').
\end{align}
Also from $\bs^{(j)}\preceq \bt^{(j)}$, the orthant submodularity implies
\begin{equation}
\label{03}
y_i^{(j)}\geq a_i^{(j)} \qquad (i\in [k]).
\end{equation}

Applying the analysis in~\cite{ity, Ward:2014wja}, we have the following.
\begin{lemma}\label{lem:1}
  Let $c\in \mathbb{R}_+$.
  Conditioning on $\bs^{(j-1)}$, suppose that
  \begin{align}
  \sum_{i=1}^k (a_{i^*}^{(j)}-a_i^{(j)})p_i^{(j)}\leq c(\sum_{i=1}^k y_i^{(j)} p_i^{(j)} )  \label{1}
  \end{align}
  holds for each $j$ with $1\leq j\leq n$,
  where $i^*=\bo(e^{(j)})$.
  Then $\mathbb{E}[f(\bs)]\geq \frac{1}{1+c}f(\bo)$.
\end{lemma}
\begin{proof}
  Conditioning on $\bs^{(j-1)}$, we have
  $\mathbb{E}[f(\bo^{(j-1)})-f(\bo^{(j)})]=\sum_i (a_{i^*}^{(j)}-a_i^{(j)})p_i^{(j)}$
  and $\mathbb{E}[f(\bs^{(j)})-f(\bs^{(j-1)})]=\sum_i y_i^{(j)}p_i^{(j)}$.
  Hence, by (\ref{1}), we have  $\mathbb{E}[f(\bo^{(j-1)})-f(\bo^{(j)})]\leq c \mathbb{E}[f(\bs^{(j)})-f(\bs^{(j-1)})]$
  (without conditioning on $\bs^{(j-1)}$).
  Note also that $\bo^{(0)}=\bo$ and $\bo^{(n)}=\bs$ by definition.
  Hence
  \begin{align*}
  f(\bo)-\mathbb{E}[f(\bs)]&=\sum_{j=1}^n \mathbb{E}[f(\bo^{(j-1)})-f(\bo^{(j)})] \\
  &\leq c(\sum_{j=1}^n \mathbb{E}[f(\bs^{(j)})-f(\bs^{(j-1)})]) \\
  &=c(\mathbb{E}[f(\bs)]-f({\bf 0})) \leq c \mathbb{E}[f(\bs)],
  \end{align*}
  and we get the statement.
\end{proof}

\subsection{A $\frac{1}{2}$-approximation algorithm for non-monotone $k$-submodular functions}\label{sec:general}
In this section, we show a polynomial-time randomized $\frac{1}{2}$-approximation algorithm for maximizing $k$-submodular functions. 
Our algorithm is described in Algorithm~\ref{alg:non-monotone}.
\begin{algorithm}
  \caption{}\label{alg:non-monotone}
  \begin{algorithmic}
  \REQUIRE{A nonnegative $k$-submodular function $f:\{0,1,\ldots,k\}^V \to \bbR_+$.}
  \ENSURE{A vector $\bs\in \{0,1,\dots, k\}^V$.}
  \STATE{$\bs\leftarrow{\bf 0}$.}
  \FOR{each $e \in V$}
      \STATE{$y_i \leftarrow \Delta_{e,i}f(\bs)$ for $1\leq i\leq k$.}
    \STATE{Assume $y_1\geq y_2 \geq \dots \geq y_k$.}
    \STATE{$i^+ \leftarrow
    \begin{cases}
     \text{the maximum integer $i$ such that $y_i> 0$} & \text{if $y_1>0$}, \\
     0 & \text{otherwise}.
     \end{cases}$
     }
    \IF{$i^+\leq 1$}
    	\STATE{$p_i\leftarrow 
    \begin{cases} 
    1 & \text{if $i=1$} \\ 
    0 & \text{otherwise}
    \end{cases} \qquad (1\leq i\leq k).
    $}
%
    \ELSIF{$i^+ =2$}
    	\STATE{$p_i\leftarrow 
    \begin{cases} 
    \frac{y_i}{y_1+y_2} & \text{if $i\in \{1,2\}$} \\ 
    0 & \text{otherwise}
    \end{cases} \qquad (1\leq i\leq k).
    $}
    \ELSE{}
    \STATE{$p_i\leftarrow 
    \begin{cases} 
    (\frac{1}{2})^i & \text{if $i\leq i^+-1$} \\ 
    (\frac{1}{2})^{i^+-1} & \text{if $i=i^+$} \\ 
    0 & \text{otherwise}
    \end{cases} \qquad (1\leq i\leq k).
    $}
    \ENDIF
    \STATE{Let $\bs(e) \in \{1,\ldots,k\}$ be chosen randomly, with $\Pr[\bs(e) = i] = p_i$ for all $i \in \{1,\ldots,k\}$.}
    \ENDFOR
  \RETURN{$\bs$}
  \end{algorithmic}
\end{algorithm}


\begin{theorem}
\label{thm:non-monotone}
Let $\bo$ be a maximizer of a $k$-submodular function $f$ and let $\bs$ be the output of Algorithm~\ref{alg:non-monotone}.
Then  $\mathbb{E}[f(\bs)]\geq  \frac{1}{2}f(\bo)$.
\end{theorem}
\begin{proof}
By Lemma~\ref{lem:1} it suffices to prove (\ref{1}) for every $1\leq j\leq n$ for $c=1$.
For simplicity of the description we shall omit the superscript $(j)$ if it is clear from the context.
Our goal is to show  
\begin{equation}
\label{2}
\sum_{1\leq i\leq k} (y_i+a_i)p_i \geq a_{i^*},
\end{equation}
which is equivalent to (\ref{1}) with $c=1$.
Recall that $y_i+y_{i'}\geq 0$ and $a_i+a_{i'}\geq 0$ for $i, i'\in [k]$ with $i\neq i'$, and 
$y_i\geq a_i$ for $i\in [k]$ (c.f.~(\ref{01}), (\ref{02}) and (\ref{03})).

If $i^+\leq 1$, then we need to show $a_1+y_1\geq a_{i^*}$. 
Since $y_i+y_{i'}\geq 0$ for $i, i'\in [k]$, we have $y_1\geq 0$.
Hence $a_1+y_1\geq a_{i^*}$ holds if $i^*=1$.
If $i^*\neq 1$ then $0\geq y_{i^*}\geq a_{i^*}$, and hence $a_1\geq 0$ by $a_1+a_{i^*}\geq 0$.
This implies $a_1+y_1\geq 0\geq a_{i^*}$.

If $i^+=2$, we need to show $(a_1+y_1)y_1+(a_2+y_2)y_2\geq a_{i^*}(y_1+y_2)$.
Now $(a_1+y_1)y_1+(a_2+y_2)y_2=a_1y_1+a_2y_2+(y_1-y_2)^2+2y_1y_2\geq a_1y_1+a_2y_2+2y_1y_2$.
If $i^*=1$, then $a_1y_1+a_2y_2+2y_1y_2\geq a_1(y_1+y_2)+(a_2+a_1)y_2\geq a_1(y_1+y_2)$
as required.
By a symmetric calculation the claim follows if $i^*=2$.
If $i^*\geq 3$, then $0\geq y_{i^*}\geq a_{i^*}$, and hence $a_1\geq 0, a_2\geq 0$.
We thus have $(a_1+y_1)y_1+(a_2+y_2)y_2\geq 0\geq a_{i^*}(y_1+y_2)$.

Hence assume $i^+\geq 3$.
Note that 
\begin{equation}
\label{eq:a-1}
y_i\geq y_{i^*}\geq a_{i^*} \qquad \text{for $i\leq i^*$}.
\end{equation}
Let $r\in {\rm argmin}\{a_i\mid i\in [k]\}$. Such $r$ is unique if $a_r<0$.

If $r=i^*$, we have $\sum_i a_ip_i\geq a_{i^*}(\sum_i p_i)=a_{i^*}$.
Since $\sum_i y_ip_i\geq 0$, (\ref{2}) follows.
Hence we assume $r\neq i^*$.

If $i^*\geq i^+$, we have $\sum_i y_ip_i=\sum_{i\leq i^+} y_ip_i\geq \sum_{i\leq i^+} a_{i^*}p_i=a_{i^*}$ by (\ref{eq:a-1})
and $\sum_i a_ip_i=\sum_{i\neq r} a_i p_i+a_rp_r\geq 0$ by $\sum_{i\neq r} p_i\geq p_r$
and $a_i+a_r\geq 0$ for $i\neq r$.
Therefore (\ref{2}) holds.
We thus assume $i^*<i^+$.

Now  we have
\begin{align}
\nonumber
\sum_i (y_i+a_i)p_i&\geq \sum_{i\leq i^*} a_{i^*}p_i+\sum_{i>i^*}a_ip_i+\sum_i a_ip_i \\
&=\left(\sum_{i<i^*} a_{i^*} p_i + 2a_{i^*}p_{i^*}\right)+\left(\sum_{i>i^*}a_ip_i+\sum_{i\neq r, i^*}a_ip_i +a_rp_r\right).
\label{3}
\end{align}
For the first term we have
\begin{align*}
\sum_{i<i^*} a_{i^*} p_i + 2a_{i^*}p_{i^*}&=a_{i^*}\left(\sum_{i<i^*}p_i+2p_{i^*}\right) 
=a_{i^*}\left(1-\left(\frac{1}{2}\right)^{i^*-1}+2\cdot \left(\frac{1}{2}\right)^{i^*}\right)=a_{i^*}.
\end{align*}
Hence it suffices to show that the second term of (\ref{3}) is nonnegative.
This is trivial if $a_r\geq 0$. Hence assume $a_r<0$.
Since $i^*< i^+$, we have
\begin{equation}
\label{eq:a-3}
\sum_{i>i^*} p_i+\sum_{i\neq r, i^*} p_i=\left(\frac{1}{2}\right)^{i^*}+\sum_{i\neq r, i^*} p_i
=p_{i^*}+\sum_{i\neq r, i^*} p_i
=1-p_r.
\end{equation}
Therefore, if $r<i^*$,
we get
\begin{equation*}
\sum_{i>i^*}a_ip_i+\sum_{i\neq r, i^*}a_ip_i+a_rp_r\geq
a_r\left(p_r-\sum_{i>i^*}p_i-\sum_{i\neq r, i^*}p_i\right)
=a_r(p_r-(1-p_r))
=a_r(2p_r-1)\geq 0,
\end{equation*}
where the first inequality follows from $a_i+a_r\geq 0$ for $i\neq r$,
the second equality follows from (\ref{eq:a-3}),
and the fourth follows from $a_r<0$ and $p_r\leq 1/2$. 
Hence we further assume $r>i^*$.
Then $p_r\leq 1/4$ by $r\neq 1$ and $i^+\geq 3$. Hence, by $a_r<0$,
\begin{align*}
&\sum_{i>i^*}a_ip_i+\sum_{i\neq r,i^*}a_ip_i +a_rp_r\geq
\sum_{i>i^*, i\neq r}a_ip_i+\sum_{i\neq r, i^*}a_ip_i+2a_rp_r
\geq a_r\left(2p_r-\sum_{i>i^*, i\neq r}p_i-\sum_{i\neq r, i^*} p_i\right) \\
&=a_r(2p_r-(p_{i^*}-p_r)-(1-p_r-p_{i^*}))
=a_r(4p_r-1)\geq 0.
\end{align*}
Thus we conclude that the second term of (\ref{3}) is nonnegative and (\ref{2}) holds.
\end{proof}

\subsection{A $\frac{k}{2k-1}$-approximation algorithm for monotone $k$-submodular functions}\label{sec:monotone}

In this section, we show a polynomial-time randomized $\frac{k}{2k-1}$-approximation algorithm for maximizing monotone $k$-submodular functions.
Our algorithm is described in Algorithm~\ref{alg:monotone}.
We note that a similar algorithm and analysis appeared in~\cite{Dobzinski:2006wi} for the submodular welfare problem, which is a special case of the monotone $k$-submodular function maximization problem.

\begin{algorithm}[t]
  \caption{}\label{alg:monotone}
  \begin{algorithmic}
  \REQUIRE{A monotone $k$-submodular function $f:\{0,1,\ldots,k\}^V \to \bbR_+$.}
  \ENSURE{A vector $\bs\in \{0,1,\dots, k\}^V$.}
  \STATE{$\bs\leftarrow{\bf 0}$.}
  \STATE{$t \leftarrow k-1$.}
  \FOR{each $e \in V$}
      \STATE{$y_i \leftarrow \Delta_{e,i}(\bs)$ for $1\leq i \leq k$.}
    \STATE{$\beta \leftarrow \sum\limits_{i=1}^k y_i^t$}
    \IF{$\beta \neq 0$}
       	\STATE{$p_i\leftarrow \frac{y_i^t}{\beta}$ \ ($1\leq i \leq k$).}
    \ELSE{}
    	\STATE{$p_i\leftarrow
    \begin{cases}
    1 & \text{if $i=1$} \\
    0 & \text{otherwise}
    \end{cases} \qquad (1\leq i\leq k).
    $}
    \ENDIF{}
        \STATE{Let $\bs(e) \in \{1,\ldots,k\}$ be chosen randomly, with $\Pr[\bs(e) = i] = p_i$ for all $i \in \{1,\ldots,k\}$.}

  \ENDFOR{}
  \RETURN{$\bs$}
  \end{algorithmic}
\end{algorithm}


It is clear that Algorithm~\ref{alg:monotone} runs in polynomial time.
Below we consider the approximation ratio of Algorithm~\ref{alg:monotone}.
\begin{theorem}
\label{thm:monotone}
Let $\bo$ be a maximizer of a monotone nonnegative $k$-submodular function $f$ and let $\bs$ be the output of Algorithm~\ref{alg:monotone}.
Then  $\mathbb{E}[f(\bs)]\geq  \frac{k}{2k-1}f(\bo)$.
\end{theorem}
\begin{proof}
  By Lemma~\ref{lem:1} it suffices to prove~\eqref{1} for every $1\leq j\leq n$ for $c=1-\frac{1}{k}$.
  For simplicity of the description we shall omit the superscript $(j)$ if it is clear from the context.

  We first consider the case $\beta = 0$.
  Since $f$ is monotone, we have $y_i = a_i = 0$ for all $1 \leq i \leq k$.
  Hence, \eqref{1} clearly holds with $c = 1-\frac{1}{k}$.

  Now suppose $\beta > 0$.
  Our goal is to show
  \begin{align}\label{eq:monotone}
    \sum_{1 \leq i \leq k}y_i^t(a_{i^*} - a_i) \leq \Bigl(1-\frac{1}{k}\Bigr)\sum_{1\leq i\leq k} y_i^{t+1}.
  \end{align}
  If $k=1$, then~\eqref{eq:monotone} follows since $i^*=1$ and both sides are equal to zero.
  Hence we assume $k\geq 2$.
  Let $\gamma = (k-1)^{\frac{1}{t}}= t^{\frac{1}{t}}$.
  Since $f$ is a monotone $k$-submodular function,
  we have  that $a_i \geq 0$ for all $i \in \{1,\ldots,k\}$.
  Then, we have
  \begin{align}
    & \sum_{i \neq i^*}y_i^t (a_{i^*}-a_i)
    \leq
    \sum_{i\neq i^*}y_i^t a_{i^*}
    \leq
    \sum_{i \neq i^*}y_i^t y_{i^*}
    =
    \frac{1}{\gamma} \biggl( \gamma y_{i^*} \cdot \sum_{i \neq i^*}y_i^t \biggr). \label{eq:monotone-1}
  \end{align}
  From the weighted AM-GM inequality, $a^{\frac{1}{t+1}}b^{\frac{t}{t+1}} \leq\frac{1}{t+1}a + \frac{t}{t+1}b$ holds for all $a,b \geq0$.
  By setting $a = (\gamma y_{i^*})^{t+1}$ and $b = (\sum_{i \neq i^*}y_i^t)^{(t+1)/t}$, we have
  \begin{align}
    \eqref{eq:monotone-1} \leq
    \frac{1}{\gamma} \biggl( \frac{1}{t+1}(\gamma y_{i^*})^{t+1} + \frac{t}{t+1}\Bigl(\sum_{i \neq i^*}y_i^t\Bigr)^{\frac{t+1}{t}}   \biggr). \label{eq:monotone-2}
  \end{align}
  From H{\"o}lder's inequality, $\sum_i a_i \leq (\sum_i a_i^{\frac{t+1}{t}})^{\frac{t}{t+1}} (\sum_i 1^{t+1})^{\frac{1}{t+1}}$ holds for any non-negative $a_i$'s.
  By setting $a_i = y_i^t$, we have
  \begin{align*}
    \eqref{eq:monotone-2} \leq &
    \frac{1}{\gamma} \biggl( \frac{1}{t+1}(\gamma y_{i^*})^{t+1} + \frac{t}{t+1}\Bigl(\sum_{i \neq i^*}y_i^{t+1}\Bigr)  \cdot \Bigl(\sum_{i \neq i^*}1^{t+1}\Bigr)^{\frac{1}{t} }  \biggr)
    \\
    = &
    \frac{1}{\gamma} \biggl(\frac{1}{t+1} (\gamma y_{i^*})^{t+1} + \frac{t(k-1)^{1/t}}{t+1}\sum_{i \neq i^*}y_i^{t+1}  \biggr) \\
    = &
    \frac{\gamma^t}{t+1}\sum_{i}y_i^{t+1}
    = \Bigl(1-\frac{1}{k}\Bigr) \sum_{i}y_i^{t+1}.
  \end{align*}
  Thus we established (\ref{eq:monotone}) and we have $k/(2k-1)$-approximation by Lemma~\ref{lem:1}.
\end{proof}

\section{Inapproximability}\label{sec:inapprox}

As we remarked in the introduction, for a symmetric submodular function $f:2^V\rightarrow \mathbb{R}_+$, a function $g:\{0,1\dots, k\}^V\rightarrow \mathbb{R}_+$ defined by
\[
 g(X_1,\dots, X_k)=\sum_{i=1}^k f(X_i)\qquad \text{( $\{X_1,\dots, X_k\}\in \{0,\dots, k\}^V$)}
\]
is $k$-submodular.
Hence one can consider an approximation algorithm for maximizing $f$ by applying an $\alpha$-approximation algorithm for $k$-submodular functions to $g$ and then returning $X_i\in {\rm argmax}\{f(X_j)\mid j\in [k]\}$ for output $(X_1,\dots, X_k)$ of the approximation algorithm.
Let $(X_1^*,\dots, X_k^*)$ be a maximizer of $g$ and $X^*$ be a maximizer of $f$.
Since $f$ is symmetric, we have $g(X_1^*,\dots, X_k^*)\geq 2f(X^*)$.
Therefore we  have $kf(X_i)\geq \sum_j f(X_j)=g(X_1,\dots, X_k)\geq \alpha g(X_1^*,\dots, X_k^*)\geq 2\alpha f(X^*)$.
Thus it gives a $2\alpha/k$-approximation algorithm for the symmetric submodular function maximization.

It was proved by Feige, Mirrokni, and Vondr\'{a}k~\cite{Feige:2011iu} that
any approximation algorithm for symmetric submodular functions with polynomial queries cannot achieve the approximation ratio better that $1/2$.
This implies that the best approximation ratio for the $k$-submodular maximization problem is at most  $\alpha\leq k/4$.
This argument, via embedding of a symmetric submodular function to a $k$-submodular function, gives the tight approximation bound for bisubmodular function, but for $k\geq 4$ it does not give a nontrivial bound.

Instead of embedding submodular functions to $k$-submodular functions,  in this section we shall directly extend the argument of~\cite{Feige:2011iu}  and establish the following bound.

\begin{theorem}
\label{thm:hardness}
For any $\varepsilon > 0$, a $(\frac{k+1}{2k}+\varepsilon)$-approximation for the monotone $k$-submodular function maximization problem would require exponentially many queries.
\end{theorem}
\begin{proof}
For simplicity we assume that $\varepsilon$ is rational.
Let $V$ be a finite set with $n=|V|$ such that  $\varepsilon n$ is an integer.
The framework of the proof is from~\cite{Feige:2011iu} (see also~\cite{Vondrak:2009hr}) and it proceeds as follows.
We shall define a $k$-submodular function $f$ and a $k$-submodular function $g_\calP$ for each  $k$-partition $\calP=\{A_1, \dots, A_k\}$ of $V$, where a $k$-partition means a partition of $V$  into $k$ subsets.
Those functions look the same as long as queries are ``balanced'' (whose definition will be given below). Suppose $\calP$ is randomly taken in the sense that  each element is added to one of the $k$ parts uniformly at random.
Then it turns out that with high probability all queries would be balanced as long as the number of queries is polynomial in $k$ and $n$.
In particular, we cannot get any information about $\calP$.
Thus one cannot distinguish $f$ and $g_\calP$ by any deterministic algorithm with a polynomial number of queries.
Hence, by Yao's min-max principle, any (possibly, randomized) algorithm with a polynomial number of queries cannot distinguish $f$ and $g_\calP$ and cannot achieve an approximation ratio better than $\frac{\max_{\bx} f(\bx)}{\max_{\bx} g_\calP(\bx)}$, which will be  $\frac{k+1}{2k}$.

Now we define $f$ and $g_\calP$.
For $\bx=(X_1,\dots, X_k)\in \{0, 1, \dots, k\}^V$, let
$n_0(\bx)=|V\setminus\bigcup_{i=1}^k X_i|$.
We define $f:\{0,\dots, k\}^V\rightarrow \bbZ_+$ by
\begin{align*}
  f(\bx)=(k+1+2k\varepsilon)n^2-(k-1)n_0(\bx)^2-2(1+k\varepsilon)nn_0(\bx) \qquad (\bx\in \{0,\dots, k\}^V).
\end{align*}
To define $g_\calP$, take any $k$-partition $\calP= \{A_1, \dots, A_k\}$.
For $\bx=(X_1,\dots, X_k)$, let $c_{i,j}(\bx)=|X_i\cap A_j|$ for $1\leq i\leq k$ and $1\leq j\leq k$,
and let $d_j(\bx)=\sum_{i=1}^k c_{i,j+i-1}$ for $1\leq j\leq k$, where the index is taken modulo $k$ ($0$ is regarded as $k$).
Then
$g_\calP:\{0,\dots, k\}^V\rightarrow \bbZ_+$ is defined by
\begin{align*}
  g_\calP(\bx)=f(\bx)+\sum_{1\leq a<b\leq k} h_\calP^{a,b}(\bx) \qquad (\bx\in \{0,\dots, k\}^V),
\end{align*}
where
\begin{align*}
  h_\calP^{a,b}(\bx):=(\max\{|d_a(\bx)-d_b(\bx)|-\varepsilon n,0\})^2.
\end{align*}
The properties of $f$ and $g_\calP$ are listed in the following claims.
\begin{claim}\label{claim:hard1}
  $f$ and $g_\calP$ for each $k$-partition $\calP$ of $V$ are nonnegative  monotone $k$-submodular functions.
\end{claim}
\begin{proof}
Clearly they are nonnegative.
To see the monotonicity and $k$-submodularity of $f$, take any $\bx$ and $e\notin \mathrm{supp}(\bx)$.
Then
\[
  \Delta_{e,i} f(\bx)=2(k-1)n_0(\bx)+2(1+k\varepsilon)n-1.
\]
This is clearly nonnegative and hence $f$ is monotone.
Also, since $n_0(\cdot)$ is non-increasing, $f$ is orthant submodular and hence $f$ is $k$-submodular by Theorem~\ref{thm:char}.

Next we consider $g_\calP$.
Take any $\bx$ and $e\notin \mathrm{supp}(\bx)$,
and suppose that $e\in A_j$.
When adding $e$ into $X_i$, $c_{i,j}(\bx)$ increases by one and hence $d_{j-i+1}(\bx)$ increases by one.
Hence for $1\leq a<b\leq k$ we have
\begin{align*}
\Delta_{e,i} h_\calP^{a,b}(\bx)
=\begin{cases}
2((d_a(\bx)-d_b(\bx))-\varepsilon n)+1 & \text{if $a=j-i+1$ and $d_a(\bx)-d_b(\bx)\geq \varepsilon n$} \\
-2((d_b(\bx)-d_a(\bx))-\varepsilon n)+1 & \text{if $a=j-i+1$ and $d_b(\bx)-d_a(\bx)\geq \varepsilon n$} \\
-2((d_a(\bx)-d_b(\bx))-\varepsilon n)+1 & \text{if $b=j-i+1$ and $d_a(\bx)-d_b(\bx)\geq \varepsilon n$} \\
2((d_b(\bx)-d_a(\bx))-\varepsilon n)+1 & \text{if $b=j-i+1$ and $d_b(\bx)-d_a(\bx)\geq \varepsilon n$} \\
0 & \text{otherwise}.
\end{cases}
\end{align*}
Hence
\begin{align*}
&\Delta_{e,i} g_\calP(\bx)=2(k-1)n_0(\bx)+2(1+k\varepsilon)n-1+\sum_{1\leq b\leq k, b\neq j-i+1} \Delta_{e,i} h_\calP^{j-i+1, b}(\bx).
\end{align*}
To see the orthant submodularity of $g_\calP$, observe that
$2n_0+\Delta_{e,i} h_\calP^{j-i+1,b}$ is non-increasing for each $b$ with $b\neq j-i+1$.
Since
\[\Delta_{e,i} g_\calP(\bx)=2(1+k\varepsilon)n-1+\sum_{1\leq b\leq k, b\neq j-i+1} (2n_0(\bx)+\Delta_{e,i} h_\calP^{j-i+1, b}(\bx)),\]
$\Delta_{e,i} g_\calP$ is non-increasing, implying the orthant submodularity.

To see the monotonicity, let $B_+=\{b\mid d_{j-i+1}(\bx)-d_b(\bx)\geq \varepsilon n\}$ and
$B_-=\{b\mid d_{b}(\bx)-d_{j-i+1}(\bx)\geq \varepsilon n\}$.
Note also that $\sum_{1\leq b\leq k} d_b(\bx)=\sum_{1\leq a\leq k, 1\leq b\leq k}c_{a,b}(\bx)=\sum_{1\leq a\leq k, 1\leq b\leq k}|X_a\cap A_b|\leq n$.
Hence
\begin{align*}
&\Delta_{e,i} g_\calP(\bx)\\
&= 2(k-1)n_0(\bx)+2(1+k\varepsilon)n-1
+\sum_{b\in B_+\cup B_-} 2(d_{j-i+1}(\bx)-d_b(\bx))+(|B_-|-|B_+|) \varepsilon n+k-1 \\
&\geq 2\left(n-\sum_{b\in B_+\cup B_-} d_b(\bx)\right)+2(k-|B_+|) \varepsilon n\geq 0.
\end{align*}
This completes the proof of Claim~\ref{claim:hard1}.
\end{proof}

\begin{claim}
\label{claim:hard2}
$\max_{\bx} f(\bx)=(k+1+2k\varepsilon)n^2$ and
$\max_{\bx} g_\calP(\bx)\geq 2kn^2(1-O(\varepsilon))$.
\end{claim}
\begin{proof}
Since $f$ is $k$-submodular, the maximum is attained for a $k$-partition by Proposition~\ref{prop}, i.e., $n_0(\bx)=0$.
Hence the maximum value of $f$ is $(k+1+2k\varepsilon)n^2$.

To see the second statement, take $\bx=(X_1,\dots, X_k)$ such that $X_i=A_i$ for $1\leq i\leq k$.
Then $d_1(\bx)=n$ and $d_j(\bx)=0$ for $2\leq j\leq n-1$,
and thus $\max_{\bx} g_\calP(\bx)\geq (k+1)n^2+(k-1)(n-\varepsilon n)^2$.
\end{proof}

Now take a random $k$-partition $\calP$ of $V$,
and consider any deterministic algorithm that tries to distinguish $f$ and $g_\calP$,
where the algorithm do not know $\calP$.
The algorithm issues some queries to the value oracle.
Call a query to $f(\bx)$ \emph{unbalanced} if $|d_i(\bx)-d_j(\bx)| \geq \varepsilon n$ for some $i,j\in [k]$,
and otherwise \emph{balanced}.
Note that $d_i(\bx)-d_j(\bx)$ can be seen as a sum of independent random variables $\{Z_e\}_{e \in V}$, where $Z_e = 1$ if $d_i(\bx)$ is increased due to $e$, that is, $e \in X_k \cap A_{i+k-1}$ for some $k$, $Z_e = -1$ if $d_j(\bx)$ is increased due to $e$, that is, $e \in X_k \cap A_{j+k-1}$ for some $k$, and $Z_e = 0$ otherwise.
By Hoeffding's inequality\footnote{Let $X_1,\dots, X_n$ be independent random variables in $[-1,1]$,
and let $\bar{X}=\frac{1}{n}(X_1+\cdots+X_n)$.
Then $\mathbb{P}(|\bar{X}-\mathbb{E}(\bar{X})|\geq t)\leq 2{\rm e}^{-\frac{nt^2}{2}}$},
the probability that $|d_i(\bx)-d_j(\bx)| \geq \varepsilon n$ for a query is at most $2{\rm e}^{-\frac{n\varepsilon^2}{2}}$.
Hence by the union bound the probability that
a query is unbalanced is at most $k^2{\rm e}^{-\frac{n\varepsilon^2}{2}}$.
Therefore, for any fixed sequence of ${\rm e}^{\frac{n\varepsilon^2}{4}}$ queries,
the probability that a query is unbalanced is still at most
${\rm e}^{\frac{n\varepsilon^2}{4}}k^2{\rm e}^{-\frac{n\varepsilon^2}{2}}=k^2{\rm e}^{-\frac{n\varepsilon^2}{4}}$.
Therefore, with probability at least $1-k^2{\rm e}^{-\frac{n\varepsilon^2}{4}}$,
all the queries will be balanced.
As long as queries are balanced, the algorithm gets the same answer regardless of $\calP$,
and it will never find out any information about the $k$-partition $\calP$.
In other words, with high probability, the algorithm will never distinguish between $f$ and $g_\calP$.
However, the maximum of $f$ is $(1+k+2k\varepsilon)n^2$ while the maximum of $g_\calP$ is at least $2k(1-O(\varepsilon))n^2$.
This means that there is no polynomial-query algorithm with approximation ratio better than $\frac{1+k}{2k}$.
\end{proof}

\section{Approximation algorithms for skew-bisubmodular functions}\label{sec:a-bisubmodular}
In this section, we discuss the problem of maximizing an $\alpha$-bisubmodular function.
An adaptation of the greedy algorithm is shown to achieve the approximation ratio of $\frac{2\sqrt{\alpha}}{(1+\sqrt{\alpha})^2}$
for $\alpha \in [0,1]$. This ratio converges to zero as $\alpha$ goes to zero. In order to improve the performance for small $\alpha$,
we give another simple approximation algorithm that achieves the approximation ratio of $\frac{1}{3+2\alpha}$
in Section~\ref{subsec:simple-approximation}. By taking the maximum of the outputs of these two algorithms,
we obtain the approximation ratio of $\frac{8}{25}$ for any $\alpha \in [0,1]$
(the minimum of the two ratio is achieved when $\alpha = \frac{1}{16}$).

Concerning the maximum of an $\alpha$-bisubmodular function, we have the following counterpart of Proposition~\ref{prop}.
\begin{lemma}\label{clm:maximal-solution}
For any $\alpha$-bisubmodular function $f:3^V\to\mathbb{R}_+$ with $\alpha\in [0,1]$,
there exists a partition of $V$ that attains the maximum value of $f$.
\end{lemma}
\begin{proof}
Suppose that $(S,T)\in 3^V$ attains the maximum value of $f$.
By the $\alpha$-bisubmodularity of $f$, we have
$$\alpha f(S,V\setminus S)+f(V\setminus T,T)\geq (1+\alpha)f(S,T),$$
which implies that $f(S,V\setminus S)=f(V\setminus T,T)=f(S,T)$.
Thus the maximum value of $f$ is attained by a partition of $V$.
\end{proof}

\subsection{A randomized greedy algorithm}\label{subsec:double-greedy}
We now extend the randomized greedy algorithm for the bisubmodular function (i.e., Algorithm~\ref{alg:non-monotone} for $2$-submodular functions).
Intuitively, $\alpha$-bisubmodularity is a variant of bisubmodularity directed toward the first argument
by parameter $\alpha$. Following this intuition, we shall adjust the choice probability as shown in Algorithm~\ref{alg:a-bisubmodular}.

\begin{algorithm}
  \caption{}\label{alg:a-bisubmodular}
  \begin{algorithmic}
  \REQUIRE{A nonnegative $\alpha$-submodular function $f:\{0,1,2\}^V \to \bbR_+$.}
  \ENSURE{A vector $\bs\in \{0,1,2\}^V$.}
  \STATE{$\bs\leftarrow{\bf 0}$.}
  \FOR{each $e \in V$}
      \STATE{$y_i \leftarrow \Delta_{e,i}f(\bs)$ for $i=1,2$.}
    \IF{$y_2<0$}
    	\STATE{$p_i\leftarrow 
    \begin{cases} 
    1 & \text{if $i=1$} \\ 
    0 & \text{if $i=2$}.
    \end{cases}
    $}
%
    \ELSIF{$y_1<0$}
      \STATE{$p_i\leftarrow 
    \begin{cases} 
    0 & \text{if $i=1$} \\ 
    1 & \text{if $i=2$}.
    \end{cases}
    $} 
    \ELSE{}
    	\STATE{$p_i\leftarrow 
    \begin{cases} 
    \frac{\alpha y_1}{\alpha y_1+y_2} & \text{if $i=1$} \\ 
    \frac{y_2}{\alpha y_1+y_2} & \text{if $i=2$}. 
    \end{cases} 
    $}
    \ENDIF
    \STATE{Let $\bs(e) \in \{1,2\}$ be chosen randomly, with $\Pr[\bs(e) = i] = p_i$ for all $i \in \{1,2\}$.}
    \ENDFOR
  \RETURN{$\bs$}
  \end{algorithmic}
\end{algorithm}

Note that, by the
$\alpha$-bisubmodularity of $f$, we have
\begin{equation}
\label{eq:a-bisub-singleton}
\alpha \Delta_{e,1}f(\bx)+\Delta_{e,2}f(\bx)\geq 0
\end{equation}
for any $\bx\in 3^V$ and $e\notin {\rm supp}(\bx)$, which implies $\alpha y_1+y_2 \geq 0$.

The following theorem provides a performance analysis of this algorithm.
\begin{theorem}\label{thr:maximize-a-bisubmodular-greedy}
For any $\alpha\in [0,1]$, the randomized greedy algorithm for maximizing $\alpha$-bisubmodular functions provides
an approximate solution within a factor of $\frac{2\sqrt{\alpha}}{(1+\sqrt{\alpha})^2}$.
\end{theorem}
\begin{proof}
Note that Lemma~\ref{lem:1} holds for any function on $\{0,1,\dots, k\}^V$ as long as the maximum is taken by a partition of $V$, which is the case in the $\alpha$-bisubmodular function maximization by Lemma~\ref{clm:maximal-solution}. Hence it suffices to prove (\ref{1}) for every $1\leq j\leq n$ for $c=(1+\alpha)/(2\sqrt{\alpha})$.
For simplicity of the description we shall omit the subscript $(j)$. Our goal is to show
\begin{equation}
\label{eq:bi-1}
(a_{i^*}-a_1)p_1+(a_{i^*}-a_2)p_2\leq c(y_1p_1+y_2p_2)
\end{equation}
for each $i^*\in \{1,2\}$ by using
\begin{align*}
\alpha a_1+a_2&\geq 0 \\
y_i&\geq  a_i \qquad (i=1,2),
\end{align*}
which follow from the $\alpha$-bisubmodularity.

If $y_1<0$, then $a_1<0$ and $y_2\geq a_2\geq -\alpha a_1\geq 0$, and hence
$cy_2\geq y_2\geq \min\{a_1-a_2, 0\}$, implying (\ref{eq:bi-1}).
A symmetric argument also implies (\ref{eq:bi-1}) if $y_2<0$.

By $\alpha y_1+y_2\geq 0$, the remaining case is when $y_1\geq 0$ and $y_2\geq 0$.
If $i^*=1$,
then 
\begin{align*}(a_1-a_2)p_2&=\frac{(a_1-a_2)y_2}{\alpha y_1+y_2}
\leq \frac{(1+\alpha)a_1y_2}{\alpha y_1+y_2}\leq \frac{(1+\alpha)y_1y_2}{\alpha y_1+y_2}
= \frac{(1+\alpha)(\alpha y_1^2+y_2^2-(\sqrt{\alpha} y_1-y_2)^2)}{2\sqrt{\alpha}(\alpha y_1+y_2)} \\
&\leq \frac{(1+\alpha)(\alpha y_1^2+y_2^2)}{2\sqrt{\alpha}(\alpha y_1+y_2)}
=\frac{1+\alpha}{2\sqrt{\alpha}}(y_1p_1+y_2p_2).\end{align*}
Thus (\ref{eq:bi-1}) holds.
On the other hand, if $i^*=2$, 
\begin{align*}(a_2-a_1)p_1&=\frac{(a_2-a_1)\alpha y_1}{\alpha y_1+y_2}
\leq \frac{(1+\alpha)a_2y_1}{\alpha y_1+y_2}\leq \frac{(1+\alpha)y_1y_2}{\alpha y_1+y_2}
= \frac{(1+\alpha)(\alpha y_1^2+y_2^2-(\sqrt{\alpha} y_1-y_2)^2)}{2\sqrt{\alpha}(\alpha y_1+y_2)} \\
&\leq \frac{(1+\alpha)(\alpha y_1^2+y_2^2)}{2\sqrt{\alpha}(\alpha y_1+y_2)}
=\frac{1+\alpha}{2\sqrt{\alpha}}(y_1p_1+y_2p_2).\end{align*}
Thus (\ref{eq:bi-1}) holds.
\end{proof}

\subsection{The second algorithm}\label{subsec:simple-approximation}
In this section, we describe another algorithm for maximizing $\alpha$-bisubmodular functions,
which achieves a better approximation ratio than the randomized greedy algorithm for small $\alpha$.

%

For an $\alpha$-bisubmodular function $f:3^V \to \mathbb{R}_+$ with $\alpha \in [0,1]$,
we define $f':2^V \to \mathbb{R}_+$ by $f'(X) = f(X,\emptyset)$ for $X\in 2^V$.
Since $f'(X)$ is a non-negative submodular function, we can apply the randomized
double greedy algorithm of \cite{bfns} to obtain a $\frac{1}{2}$-approximate solution $Z$
to the maximization of $f'$. Our second algorithm for $\alpha$-bisubmodular function
maximization is rather simple: Take the better of $(\emptyset,V)$ and $(Z,\emptyset)$.

\begin{theorem}\label{thr:maximize-a-bisubmodular-simple}
The second algorithm for $\alpha$-bisubmodular function maximization problem
provides a $\frac{1}{3+2\alpha}$-approximate solution for any $\alpha\in [0,1]$.
%
\end{theorem}

\begin{proof}
  Let $(S,T)$ be an optimal solution.
  By the $\alpha$-bisubmodularity, we have
  \begin{align*}
    f(S,\emptyset) + f(\emptyset,V)
    & \geq
    f(\emptyset,\emptyset) + \alpha f(\emptyset,T) + (1-\alpha) f(S,T),\\
    f(S,\emptyset) + f(\emptyset,T)
    & \geq
    f(\emptyset,\emptyset) + f(S,T),
  \end{align*}
  which imply
  \begin{align*}
    (1+\alpha)f(S,\emptyset) + f(\emptyset,V)    \geq  (1+\alpha) f(\emptyset,\emptyset) + f(S,T) \geq f(S,T).\label{eq:simple-opt}
  \end{align*}

  Let $(A,B)$ be the output by the algorithm.
  Then,
  \begin{align*}
    &\mathbb{E}[f(A,B)]
    \geq
    \max\left\{ \frac{1}{2}f(S,\emptyset),f(\emptyset,V) \right\}
    \geq
    \frac{2+2\alpha}{3+2\alpha} \cdot \frac{1}{2}f(S,\emptyset) + \frac{1}{3+2\alpha} f(\emptyset,V)  \\
    = &
    \frac{1}{3+2\alpha} \left((1+\alpha) f(S,\emptyset) + f(\emptyset,V)\right)
    \geq \frac{1}{3+2\alpha} f(S,T).
      \qedhere
  \end{align*}
\end{proof}

Combining this with the first algorithm, we obtain an approximate solution within a factor of
$\max \left\{\frac{2\sqrt{\alpha}}{(1+\sqrt{\alpha})^2}, \frac{1}{3+2\alpha} \right\}$.
The minimum of this ratio is $\frac{8}{25}$, which is achieved when $\alpha = \frac{1}{16}$.

\bibliographystyle{abbrv}
\bibliography{k-submodular}
\end{document}